\documentclass[12pt, draftclsnofoot,onecolumn]{IEEEtran}

\usepackage{graphicx,subcaption,xcolor,setspace,amsthm,cite,booktabs,epstopdf,dblfloatfix,blindtext,amsthm}
\usepackage{algorithm,algcompatible,multirow}

\usepackage[binary-units,detect-weight=true, detect-family=true]{siunitx}
\DeclareSIUnit \bitspersecond {bps}
\newcommand{\centrallocation}{/Users/AnumAli/Dropbox/STUDY/Central} 
\newtheorem{thm}{Theorem}
\newtheorem{prop}{Proposition}

\usepackage{amsmath,amssymb}

\newcommand{\bbC}{{\mathbb{C}}}

\newcommand{\bbE}{{\mathbb{E}}}

\newcommand{\bbR}{{\mathbb{R}}}

\newcommand{\bg}{{\mathbf{g}}}
\newcommand{\bh}{{\mathbf{h}}}

\newcommand{\bx}{{\mathbf{x}}}
\newcommand{\by}{{\mathbf{y}}}
\newcommand{\bz}{{\mathbf{z}}}
\newcommand{\bzero}{{\mathbf{0}}}
\newcommand{\bone}{{\mathbf{1}}}
\newcommand{\bA}{{\mathbf{A}}}
\newcommand{\bB}{{\mathbf{B}}}

\newcommand{\bD}{{\mathbf{D}}}
\newcommand{\bE}{{\mathbf{E}}}

\newcommand{\bG}{{\mathbf{G}}}
\newcommand{\bH}{{\mathbf{H}}}
\newcommand{\bI}{{\mathbf{I}}}

\newcommand{\bP}{{\mathbf{P}}}

\newcommand{\bR}{{\mathbf{R}}}

\newcommand{\bV}{{\mathbf{V}}}

\newcommand{\bX}{{\mathbf{X}}}

\newcommand{\rmB}{{\mathrm{B}}}
\newcommand{\rmC}{{\mathrm{C}}}

\newcommand{\rmF}{{\mathrm{F}}}

\newcommand{\rmZ}{{\mathrm{Z}}}

\newcommand{\cC}{\mathcal{C}}
\newcommand{\cD}{\mathcal{D}}

\newcommand{\cN}{\mathcal{N}}
\newcommand{\cO}{\mathcal{O}}

\newcommand{\cX}{\mathcal{X}}

\newcommand{\bTheta}{\boldsymbol{\Theta}}

\newcommand{\transp}{{\sf T}}

\newcommand{\tr}{{\rm tr}}

\newcommand{\diag}{{\rm diag}}

\makeatletter
\def\munderbar#1{\underline{\sbox\tw@{$#1$}\dp\tw@\z@\box\tw@}}
\makeatother

\newcommand{\zefo}{{\rmZ\rmF}}
\newcommand{\cobe}{{\rmC\rmB}}

\newcommand{\SNR}{{\rm SNR}}
\newcommand{\SINR}{{\rm SINR}}

\begin{document}
\bstctlcite{IEEEmax3beforeetal}
\title{Linear Receivers in Non-stationary Massive MIMO Channels with Visibility Regions}
\author{Anum Ali, {\it Student Member, IEEE},  Elisabeth de Carvalho, {\it Senior Member, IEEE}, and Robert W. Heath Jr., {\it Fellow, IEEE}
\thanks{This work was supported in part by TACTILENet (Grant no. 690893), within the Horizon 2020 Program, by the Danish Council for Independent Research (Det Frie Forskningsråd) DFF-133500273, by the U.S. Department of Transportation through the Data-Supported Transportation Operations and Planning (D-STOP) Tier 1 University Transportation Center, and by the National Science Foundation under Grant No. ECCS-1711702.}
\thanks{A. Ali and R. W. Heath Jr. are with the University of Texas at Austin,
Austin, TX 78701 USA (e-mail: \{anumali,rheath\}@utexas.edu). 
}
\thanks{E. de Carvalho is with the Aalborg University, 9400 Aalborg, Denmark (e-mail:
edc@es.aau.dk).}\vspace{-2em}}

\maketitle
%
\begin{abstract}
In a massive MIMO system with large arrays, the channel becomes spatially non-stationary. We study the impact of spatial non-stationarity characterized by visibility regions (VRs) where the channel energy is significant on a portion of the array. Relying on a channel model based on VRs, we provide expressions of the signal-to-interference-plus-noise ratio (SINR) of conjugate beamforming (CB) and zero-forcing (ZF) precoders. We also provide an approximate deterministic equivalent of the SINR of ZF precoders. We identify favorable and unfavorable multi-user configurations of the VRs and compare the performance of both stationary and non-stationary channels through analysis and numerical simulations.  
\end{abstract}
\vspace{-0.5em}\section{Introduction}
A massive MIMO system~\cite{Marzetta2010Noncooperative} is characterized by the use of many antennas and support for multiple users. At the extreme, the arrays may be physically very large~\cite{Martinez2014Towards,Hu2017Massive,Truong2013viability} and integrated into large structures like stadiums, or shopping malls. Unfortunately, when the dimension of the antenna array becomes large, different kinds of non-stationarities appear across the array. Further, different parts of the array may observe the same channel paths with different power, or even entirely different channel paths~\cite{Martinez2014Towards}. This effect may even be observed for compact arrays~\cite{Martinez2014Towards}. We show that non-stationarity in massive arrays has a significant impact on performance assessment and transceiver design. 

We propose a simple non-stationary channel model and analyze the performance of CB and ZF in the downlink of a multi-user massive MIMO system. The channel model is based on VRs that capture the received power variation across the array. We also propose a closed-form approximation of the~\SINR~of the ZF receiver. The expression shows the dependence of the SINR on channel parameters and allows a comparison between spatially stationary and non-stationary channels. The analysis and simulation results show that the impact of spatial non-stationarity on the performance of linear receivers is scenario dependent.

Few theoretical studies exist on spatially non-stationary channels in massive MIMO systems. In~\cite{Zhou2015Spherical} a spherical wave-front based LOS channel model was proposed and the channel capacity is studied with the proposed model. In~\cite{Li2015Capacity} an upper bound on the ergodic capacity of a non-stationary channel was provided. No prior work, however, has studied the performance of linear precoders with non-stationary channels.

\textbf{Notation:} $\bX$ is a matrix, $\bx$ is a vector, $\cX$ is a set, $x$ and $X$ are scalars. Superscript $\transp$, and $\ast$, represent transpose, and conjugate transpose, respectively. $\bbE[\cdot]$ is the expectation, and $\cC\cN(\bx,\bX)$ is a complex Normal with mean $\bx$ and covariance $\bX$. The identity matrix is $\bI$ and $\|\bx\|_p$ is the $p$-norm. The cardinality of a set $\cX$ is $|\cX|$. The matrix $\bX=\diag(\bx)$ is a diagonal matrix with the vector $\bx$ on its main diagonal, and $\tr(\bX)$ is the trace of matrix $\bX$. The operator $\overset{\rm{a.s.}}{\rightarrow}$ denotes almost sure convergence.
\vspace{-0.5em}\section{System and channel model}
We consider a narrowband broadcast system where the base-station (BS) equipped with $M$ antennas is serving $K$ single-antenna users ($M\geq K$). The BS serves all the users using the same time-frequency resource. The signal for user $k$, $s_k$ is precoded by $\bg_k\in\bbC^M$ and scaled by the signal power $p_k\geq 0$ before transmission. The transmit vector $\bx$ is the linear combination of the precoded and scaled signals of all the users, i.e.,
\begin{align}
\bx=\sum_{k=1}^K\sqrt{p_k}\bg_ks_k.
\end{align}
Let $\bG=[\bg_1,\bg_2,\cdots,\bg_K]\in\bbC^{M\times K}$ be the combined precoding matrix, $\bP=\diag([p_1,p_2,\cdots,p_K]^\transp)\in \bbR^{K\times K}$ be the diagonal matrix of signal powers, and $P\geq 0$ be the total power. The combined precoding matrix $\bG$ is normalized to satisfy the power constraint
\begin{align}
\bbE[ \| \bx \|^2 ]=\tr(\bP \bG^\ast \bG)= P.
\label{eq:powerconst}
\end{align}

Let $\bh_k\in\bbC^{M}$ denotes the random channel from the BS to user $k$. Then the received signal at the user $k$ is 
\begin{align}
\by_k=\bh_k^\ast\bx+n_k,~k=1,2,\cdots,K,
\label{eq:rxdsgnl}
\end{align}
where $n_k\sim\cC\cN(\bzero,\sigma^2)$ is the additive noise. Assuming independent Gaussian signaling, i.e., $s_k\sim\cC\cN(0,1)$ and $\bbE[s_i s_j^\ast]=0,~i\neq j$, the \SINR~$\gamma_k$ of user $k$ can be written as 
\begin{align}
\gamma_k=\frac{p_k |\bh_k^\ast \bg_k|^2 }{\sum\limits_{j=1,j\neq k}^K p_j |\bh_k^\ast \bg_j|^2 + \sigma^2}.
\label{eq:SINR}
\end{align}

Let $\bH=[\bh_1~\bh_2~\cdots~\bh_K]\in\bbC^{M\times K}$ denote the channel matrix between the BS and $K$ users. Then, the CB precoder is
\begin{align}
\bG_{\cobe}=\beta_{\rmC\rmB} \bH,
\label{eq:CBprecoder}
\end{align}
and the ZF precoder is
\begin{align}
\bG_{\zefo}=\beta_{\rmZ\rmF} \bH(\bH^\ast\bH)^{-1},
\label{eq:ZFprecoder}
\end{align}
where the scaling factors $\beta_{\rmC\rmB}=\sqrt{P/\tr(\bP\bH^\ast\bH)}$ and $\beta_{\rmZ\rmF}=\sqrt{P/\tr(\bP(\bH^\ast\bH)^{-1})}$ ensure that the power constraint \eqref{eq:powerconst} is met.

By defining $\rho= P/\sigma^2$ as the signal-to-noise ratio (\SNR) and using~\eqref{eq:CBprecoder} in~\eqref{eq:SINR}, the \SINR~of the $k$th user for CB is
\begin{align}
\gamma_k^{(\cobe)}=p_k\frac{\rho |\bh_k^\ast\bh_k|^2}{\rho \sum\limits_{j=1,j\neq k}^K p_j |\bh_k^\ast \bh_j|^2 + \tr(\bP\bH^\ast\bH)}.
\label{eq:SINRCB}
\end{align}
Similarly, using~\eqref{eq:ZFprecoder} in~\eqref{eq:SINR}, the \SINR~of the $k$th user for ZF is
\begin{align}
\gamma_k^{(\zefo)}=p_k\frac{\rho}{\tr(\bP(\bH^\ast\bH)^{-1})}.
\label{eq:SINRZF}
\end{align}

Let $\bR_k\in\bbC^{M\times M}$ be the spatial correlation matrix of user $k$ corresponding to the case of a stationary channel. Further, let $\bD_k$ be a diagonal matrix such that if the signal transmitted from only $D_k$ antennas is received by the user $k$, $\bD_k$ has $D_k$ non-zero diagonal entries. This diagonal matrix $\bD_k$ models the VR of user $k$. For the proposed channel model, we introduce a matrix $\bTheta_k$ of the form
\begin{align}
\bTheta_k=\bD_k^{\frac{1}{2}}\bR_k\bD_k^{\frac{1}{2}}.
\label{eq:gencov}
\end{align}

If $\bz_k\sim\cC\cN(0,\frac{1}{M}\bI)$, then by the proposed model, the channel of user $k$, $\bh_k$ is
\begin{align}
\bh_k=\sqrt{M}\bTheta_k^{\frac{1}{2}}\bz_k.
\label{eq:chmod}
\end{align}
For	 stationary channel $\bD_k=\bI$, and $\bTheta_k=\bR_k$. Therefore, the channel model~\eqref{eq:chmod} subsumes the well known correlated channel studied in~\cite{Hoydis2013Massive,Wagner2012Large}. 
\vspace{-0.5em}\section{Large system analysis in stationary channels}
Correlated stationary channels (i.e., $\bTheta_k=\bR_k$) were studied in~\cite{Hoydis2013Massive,Wagner2012Large}, under the following assumptions.
\begin{enumerate}
\item[A1] $M,K,\frac{M}{K}\rightarrow\infty$. (BS equipped with a large number of antennas serving a large number of users.)
\item[A2] The covariance matrices have a uniformly bounded spectral norm i.e., $\underset{M,K\rightarrow\infty}{\lim\sup} \underset{1\leq k \leq K}{\sup} \|\bR_k\| = \cO(1)$~\cite{Bjoernson2014Massive}. (User channels are not highly correlated.)
\item[A3] The power $p_{\rm{max}}=\max(p_1,p_2,\cdots,p_K)$ is of the order $O(1/K)$, i.e., $\|\bP\|=O(1/K)$. (Transmission power for all the users is on the same order.)
\end{enumerate}

With A1-A3, the deterministic equivalent of $\gamma_k^{(\cobe)}$ can be written as~\cite[eq. 24]{Hoydis2013Massive}
\begin{align}
\bar\gamma_k^{(\cobe)}=p_k \frac{\rho(\tr(\bR_k)^2)}{\rho \sum\limits_{j=1,j\neq k}^K p_j \tr(\bR_k\bR_j) + \sum\limits_{j=1}^{K}p_j\tr(\bR_j)}.
\label{eq:CBana}
\end{align}
where $\gamma_k^{(\cobe)}-\bar\gamma_k^{(\cobe)} \overset{\rm{a.s.}}{\underset{M\rightarrow \infty}{\longrightarrow}}0$. Using A1-A3 and some additional assumptions, the deterministic equivalent of $\gamma_k^{(\zefo)}$ was obtained in~\cite[eq. 34]{Wagner2012Large}. That expression, though, is not closed form and as such is not suitable for comparing stationary and non-stationary channels. With this motivation, we provide a closed form approximate expression in the following theorem. 
\vspace{-1.5em}
\begin{thm}
Under the assumptions A1-A3, an approximate deterministic equivalent of $\gamma_k^{(\zefo)}$ in~\eqref{eq:SINRZF} is
\begin{align}
{\bar\gamma}_k^{(\zefo)}=p_k\frac{\rho} {\sum\limits_{i=1}^{K}p_i\big(\tr(\bR_i)-\sum\limits_{j=1,j\neq i}^{K}\frac{\tr(\bR_i\bR_j)}{\tr(\bR_j)}\big)^{-1}},
\label{eq:ZFanaapp}
\end{align}
which is guaranteed to be non-negative as $\frac{M}{K}\rightarrow\infty$.
\end{thm}
\begin{proof}
See Appendix A. 
\end{proof}

\textbf{Remark:} The expression~\eqref{eq:ZFanaapp} is obtained with the help of a diagonal approximation (see Appendix A). In the following proposition, we provide the order of absolute error in this approximation.
\begin{prop}
The error in approximating the $O(K)$ term $\bh_i^\ast \bar\bH_i (\bar\bH_i^\ast \bar\bH_i)^{-1} \bar\bH_i^\ast \bh_i$ by $\bh_i^\ast\bar\bH_i\bV^{-1} \bar\bH_i^\ast\bh_i$ (where $\bV$ is defined in~\eqref{eq:diagapp}), is
\begin{align}
\epsilon_i=|\bh_i^\ast \bar\bH_i (\bar\bH_i^\ast \bar\bH_i)^{-1} \bar\bH_i^\ast \bh_i-\bh_i^\ast\bar\bH_i\bV^{-1} \bar\bH_i^\ast\bh_i|,
\label{eq:epsilon}
\end{align}
and is of $O\left(\frac{K}{\sqrt{M}}\right)$. 
\end{prop}
\begin{proof}
See Appendix B.
\end{proof}
\vspace{-0.5em}
\section{Stationary versus non-stationarity channels}\label{sec:cost}
For non-stationary channels, we introduce two additional assumptions.
\begin{enumerate}
\item[A4] $D_k\rightarrow\infty,~\forall k.$ (Large VR for each user).
\item[A5] $\underset{M,K\rightarrow\infty}{\lim\sup} \underset{1\leq k \leq K}{\sup} \|\bTheta_k\| < \infty$. (Non-stationary user channels are not highly correlated.)
\end{enumerate}
With this, the same analysis that led to~\eqref{eq:CBana} and~\eqref{eq:ZFanaapp} can be used for non-stationary channels. In fact, only $\bR$ needs to be replaced with $\bTheta$ in theoretical expressions~\eqref{eq:CBana}, and~\eqref{eq:ZFanaapp} to get the results for the non-stationary case. 

For comparison between stationary and non-stationary channels, we make a few simplistic choices of system and channel parameters. Specifically, we start by considering $p_k=\frac{P}{K}~\forall k$~and $\bR_k=\bI~\forall k$. For non-stationary channels, we consider two types of channel normalization.

\textbf{Normalization 1, $\tr(\bTheta_k)=\tr(\bR_k)=M,~\forall k$:} This ensures that the stationary and non-stationary channels have the same norm. The physical implication of normalization 1 is shown in Fig.~\ref{fig:chstnonst}~(a). User terminal 1 (i.e., farther) receives signal from all antennas but with lower power, whereas user terminal 2 (i.e., closer) receives signal from fewer antennas but with higher power. This is achieved by choosing $\bD_k=\diag([\bzero, \sqrt{\frac{M}{D_k}}\bone_{D_k}, \bzero]^\transp)$.
 
\textbf{Normalization 2, $\tr(\bTheta_k)=D_k~\forall k$:} The physical implication of normalization 2 is shown in Fig.~\ref{fig:chstnonst}~(b). User terminal $1$ and $2$ are equidistant from the BS, however, user terminal $2$ receives signal from only a few antennas. This is achieved by choosing $\bD_k=\diag([\bzero, \bone_{D_k}, \bzero]^\transp)$.

To further simplify the comparison, we assume that $D_k=D~\forall k$.
\begin{figure}[h!]
\centering
\begin{subfigure}{0.38\textwidth}
        \centering
        \includegraphics[width=1\textwidth]{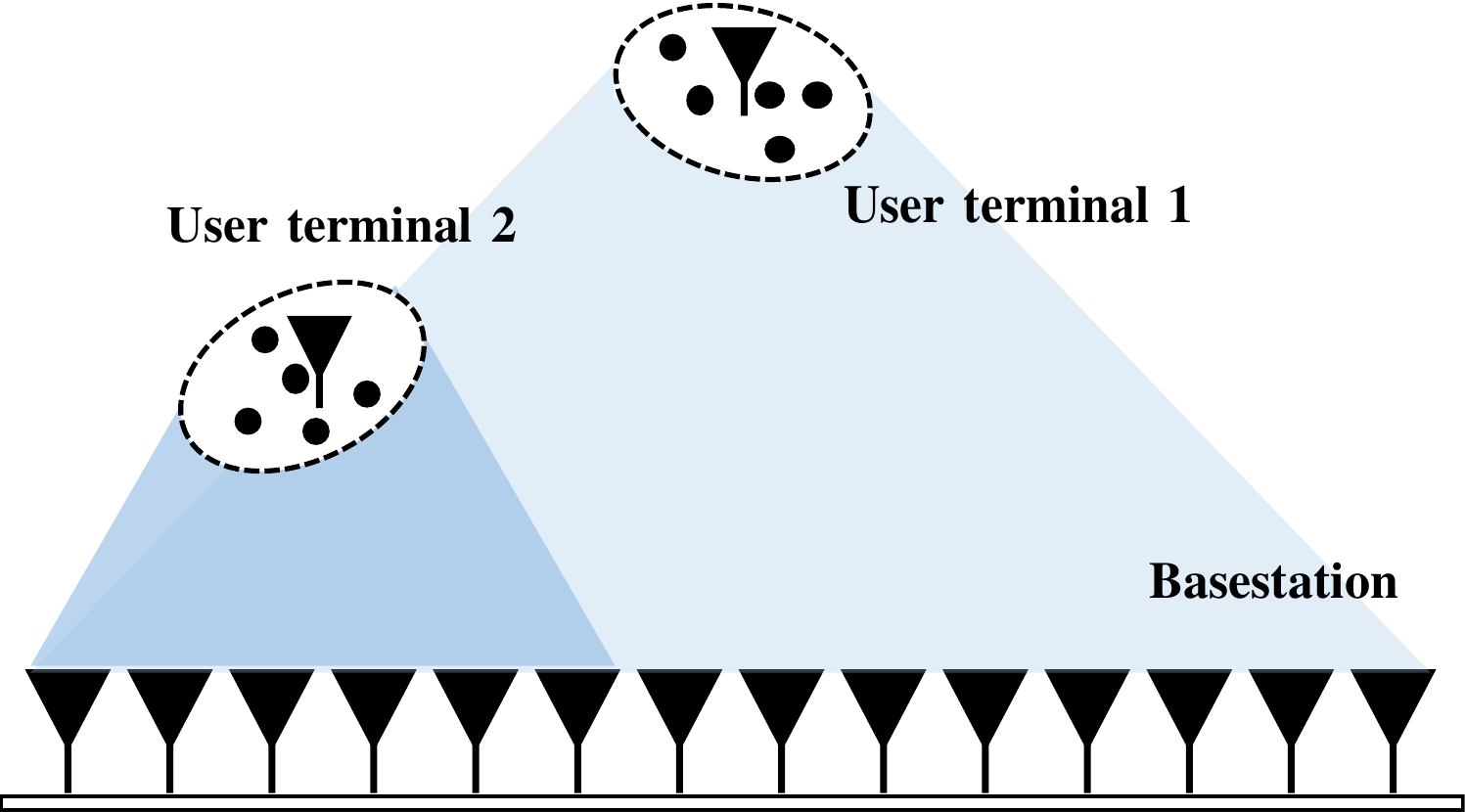}
        \caption{Physical implication of normalization 1 i.e., $\tr(\bTheta)=M$.}
        \label{fig:st_non_st1}
        \end{subfigure}
            \begin{subfigure}{0.38\textwidth}
        \centering
        \includegraphics[width=1\textwidth]{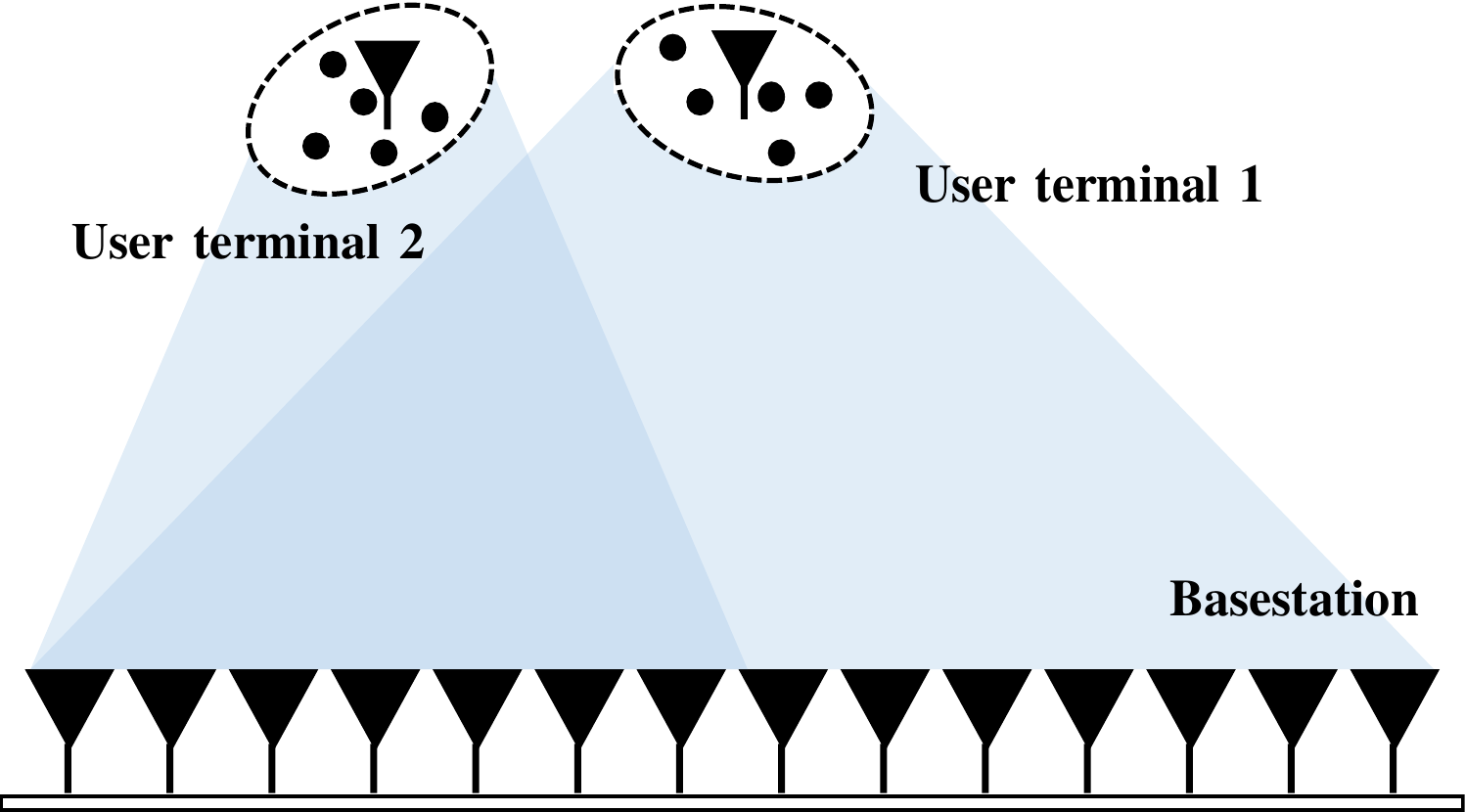}
        \caption{Physical implication of normalization 2 i.e., $\tr(\bTheta)=D$.}
        \label{fig:st_non_st2}
    \end{subfigure}
  \hfill
        \caption{Physical implication of normalization 1 (i.e., $\tr(\bTheta)=M$) and normalization 2 (i.e., $\tr(\bTheta)=D$).}
        \label{fig:chstnonst}
\end{figure}

We outline the \SINR~for CB in stationary and non-stationary channels with normalization 1 in detail below. The derivations for normalization 2 and ZF precoding are similar and the results are summarized in Table~~\ref{tab:nonstationaritycost}. 
  
The \SINR~of CB~\eqref{eq:CBana} for stationary channels simplifies to
\begin{align}
\gamma_k^{(\cobe)}-{\rm{st.}}=\frac{\rho M}{\rho(K-1)+K}.
\label{eq:CBstsim}
\end{align}

For the non-stationary channel under consideration, a user receives the signal transmitted from $D$ antennas. The indices of these antennas for user $k$ are collected in a set $\cD_k$. The \SINR~of user $k$ depends on $|\cD_k\cap\cD_j|~\forall j\neq k$ (i.e., inter-user interference). For example, if $K=2$, $D=M/2$ and $\cD_1=\cD_2=\{1,\cdots,M/2\}$, then $|\cD_1\cap\cD_2|=D$ and there is high inter-user interference. If, however, $\cD_1=\{1,\cdots,M/2\}$ and $\cD_2=\{M/2+1,\cdots,M\}$, then $|\cD_1\cap\cD_2|=0$ and there is no inter-user interference. Thus for non-stationary channels, we consider the best-case (and worst-case) index sets $\cD_k$ that result in maximum (and minimum) possible~\SINR~for the considered setup. 

In the worst-case, there is high inter-user interference. This happens, when all the $K$ users receive the signal from the same $D$ antenna elements. In this case, assuming $\tr(\bTheta)=M$, it can be shown that
\begin{align}
\gamma_k^{(\cobe)}-{\rm{non~st.~(worst)}}=\frac{\rho M}{\rho\frac{M}{D}(K-1)+K}.
\label{eq:CBNstworst}
\end{align}  
There is an additional factor $M/D$ in the first term in denominator of~\eqref{eq:CBNstworst} compared to~\eqref{eq:CBstsim}.  Therefore, for worst-case, the smaller the VR of the user (i.e., in this case, the number of active antennas $D$), the more \SINR~loss for non-stationary channels. 

In the best-case, there is low inter-user interference for all the $K$ users. The best-case antenna indices can be found using counting arguments. Asymptotically, with a user $k$ receiving signal from $D$ antennas, and a total of $M$ antennas, we can arrange only $M/D$ users without any inter-user interference. If we continue this arrangement for all users, there will be $\frac{KD}{M}-1$ interfering users for any user $k$. With this observation, the best-case \SINR~can be written as
\begin{align}
\gamma_k^{(\cobe)}-{\rm{non~st.~(best)}}=\frac{\rho M}{\rho(K-\frac{M}{D})+K}.
\label{eq:CBNstbest}
\end{align}

If $KD\leq M$, there will be no inter-user interference with the arrangement described above and~\eqref{eq:CBNstbest} can be further simplified. Note that the first term in the denominator of~\eqref{eq:CBstsim} and~\eqref{eq:CBNstbest} differs. Specifically, for best-case, if $\frac{M}{D}$ is large (i.e., smaller VR), then the \SINR~of CB precoders for non-stationary channels can be better than the stationary channels.

\begin{table*}[h!]
\centering
\caption{\SINR~expressions for CB and ZF precoders for stationary and non-stationary channels.}
\label{tab:nonstationaritycost}
\setlength\extrarowheight{5pt}
\begin{tabular}{c|c|ccc|}
\cline{2-5}
& \multirow{2}{*}{Stationary} &  \multicolumn{3}{c|}{Non-stationary}\\
\cline{3-5}
&&&Worst&Best\\
\cline{2-5}
\multirow{2}{*}{CB}&\multirow{2}{*}{$\frac{\rho M}{\rho(K-1)+K}$}&Normalization 1: $\tr(\bTheta_k)=M$&$\frac{\rho M}{\rho\frac{M}{D}(K-1)+K}$
&$\frac{\rho M}{\rho(K-\frac{M}{D})+K}$\\
\cline{3-5}
&&Normalization 2: $\tr(\bTheta_k)=D$&$\frac{\rho D}{\rho(K-1)+K}$&$\frac{\rho D}{\rho(\frac{KD}{M}-1)+K}$
\\
\cline{2-5}
\multirow{2}{*}{ZF}&\multirow{2}{*}{$\frac{\rho(M-K+1)}{K}$}&Normalization 1: $\tr(\bTheta_k)=M$&$\frac{\rho(M-\frac{M}{D}(K-1))}{K}$&$\frac{\rho(M-K+\frac{M}{D})}{K}$\\
\cline{3-5}
&&Normalization 2: $\tr(\bTheta_k)=D$&$\frac{\rho(D-K+1)}{K}$&$\frac{\rho(D-\frac{KD}{M}+1)}{K}$\\
\cline{2-5}
\end{tabular}
\end{table*}
\vspace{-0.5em}\section{Numerical results}
We verify the analysis of the non-stationary channels. We consider $M=60$, $K=M/2$, and $\rho=\SI{10}{\decibel}$. We plot the \SINR~results against the active number of antennas per user $D$. We consider both the~best-case and worst-case antenna configurations discussed in Section~\ref{sec:cost}. We plot the results for normalization 1 i.e., $\tr(\bTheta)=M$ in Fig.~\ref{fig:Res_IID}, and for normalization 2 i.e., $\tr(\bTheta)=D$ in Fig.~\ref{fig:Res_IID_NN}. From Fig.~\ref{fig:Res_IID}, notice that when $D$ is small, the \SINR~of the non-stationary channels in the best-case (worst-case) is higher (lower) than the stationary channels. The worst-case performance loss for CB (ZF) is as high as $\SI{15}{\decibel}$ ($\SI{12.5}{\decibel}$). As $D$ increases, however, the \SINR~in the non-stationary channels converge to the \SINR~of the stationary channels. This observation holds for both CB and ZF precoding. The ZF curve in the worst-case starts from $D=30$. For smaller values of $D$, the channel matrix $\bH$ is rank deficient and ZF precoding fails. From Fig.~\ref{fig:Res_IID_NN}, we notice that with normalization $2$, the \SINR~of the non-stationary channels is lower than the~\SINR~of the stationary channels (for $\rho=\SI{10}{\decibel}$). With normalization 2, the performance loss for both CB and ZF can be as high as $\SI{15}{\decibel}$.

\begin{figure}[h!]
\centering
\begin{subfigure}{0.38\textwidth}
        \centering
        \includegraphics[width=\textwidth]{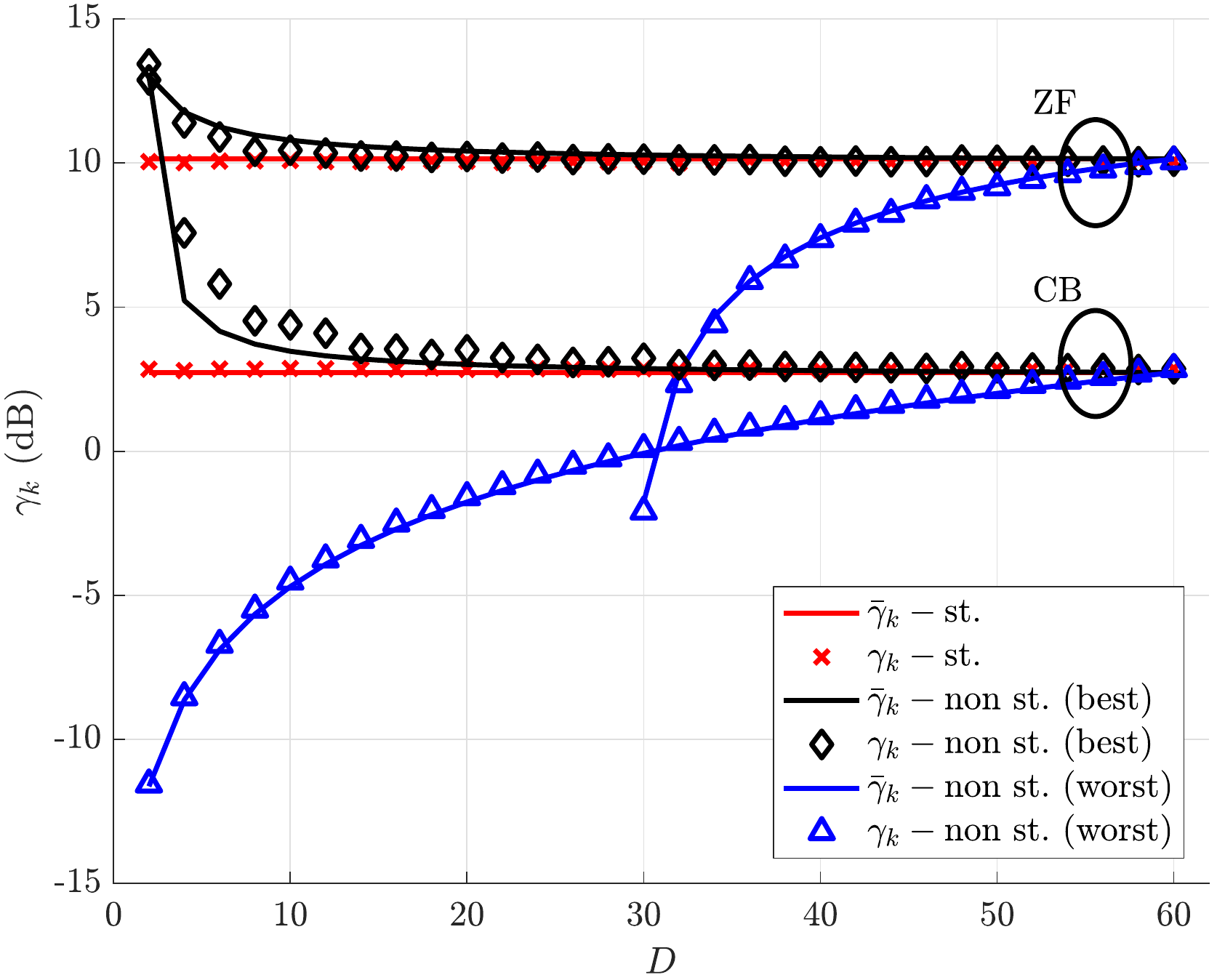}
        \caption{Normalization 1: $\tr(\bTheta)=M$.}
        \label{fig:Res_IID}
        \end{subfigure}
            \begin{subfigure}{0.38\textwidth}
        \centering
        \includegraphics[width=\textwidth]{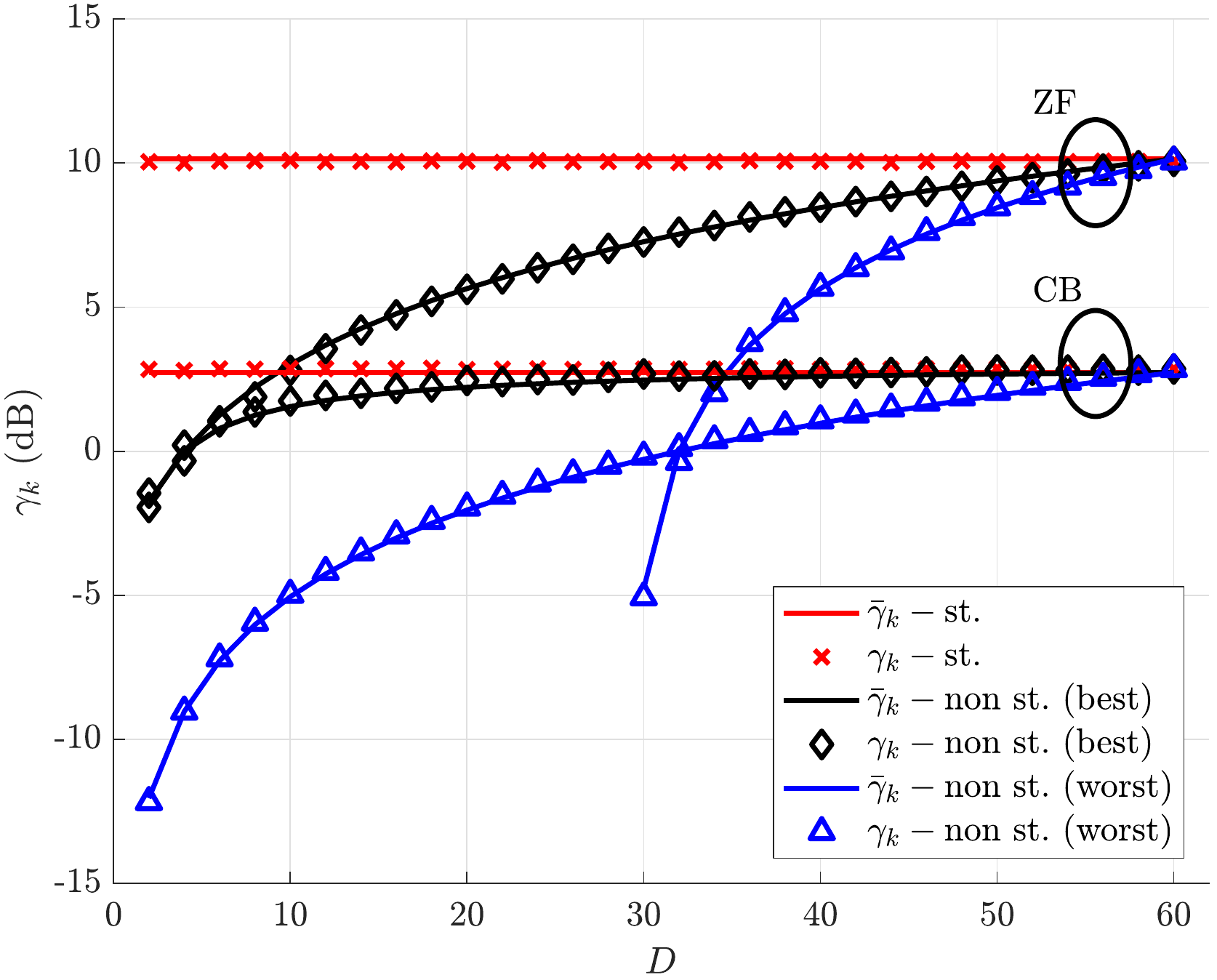}
        \caption{Normalization 2: $\tr(\bTheta)=D$.}
        \label{fig:Res_IID_NN}
    \end{subfigure}
  \hfill
        \caption{The \SINR~vs the active number of antennas $D$ ($M=60$, $K=30$, and $\rho=\SI{10}{\decibel}$).}
        \label{fig:Res}
\end{figure}

\vspace{-0.5em}\section{Conclusion}
The VR of the channel impacts the performance of CB and ZF precoders significantly. For small VRs, the post processing~\SINR~loss compared with the stationary channels can be as high as \SI{15}{\decibel} for both CB and ZF. In the best-case, i.e., when the VRs of the users reduce inter-user interference, the post-processing~\SINR~for both CB and ZF can be higher compared to a stationary channel. Finally, small VRs can make the channel rank deficient and render the ZF precoding infeasible.
\vspace{-0.5em}\section*{Appendix A: Proof of Theorem 1}
The \SINR~for ZF~\eqref{eq:SINRZF} can be re-written as
\begin{align}
\gamma_k^{\zefo}=p_k\frac{\rho}{\sum\limits_{i=1}^{K}p_i(\bH^\ast\bH)^{-1}_{i,i}},
\label{eq:SINRZFsumform}
\end{align}
where $(\bH^\ast\bH)^{-1}_{i,i}$ is the $i$th diagonal entry of the inverse matrix $(\bH^\ast\bH)^{-1}$. This entry can be re-written as
\begin{align}
(\bH^\ast\bH)^{-1}_{i,i}=(\bh_i^\ast \bh_i - \bh_i^\ast \bar\bH_i (\bar\bH_i^\ast \bar\bH_i)^{-1} \bar\bH_i^\ast \bh_i)^{-1},
\label{eq:invivalue}
\end{align}
where $\bar\bH_i=[\bh_1,~\cdots,~\bh_{i-1},~\bh_{i+1},~\cdots,~\bh_K]$. The first term on the RHS of~\eqref{eq:invivalue} can be evaluated using~\eqref{eq:chmod} and~\cite[Lemma 4]{Wagner2012Large}, i.e.,
\begin{align}
\bz_i^\ast\bR_i\bz_i-\frac{1}{M}\tr(\bR_i)\overset{\rm{a.s.}}{\underset{M\rightarrow \infty}{\longrightarrow}}0.
\label{eq:oneRHS}
\end{align}
For the second term, we approximate $(\bar\bH_i^\ast\bar\bH_i)^{-1}$ by retaining only its diagonal entries. For large $M$, approximating the off-diagonal terms to $0$ is reasonable as due to~\cite[Lemma 5]{Wagner2012Large}
\begin{align}
\bz_i^\ast\bR_i^{\frac{1}{2}}\bR_j^{\frac{1}{2}}\bz_j\overset{\rm{a.s.}}{\underset{M\rightarrow \infty}{\longrightarrow}}0.
\label{eq:crossterms}
\end{align}
We retain the diagonal entries in a matrix $\bV_i$ defined as
\begin{align}
\bV_{ij}=\begin{cases}
(\bar\bH_i^\ast\bar\bH_i)_{i,i}&\text{when }i=j,\\
0&\text{otherwise}.
\end{cases}
\label{eq:diagapp}
\end{align}

With the diagonal approximation, the second term on the RHS can be evaluated as
\begin{align}
\bh_i^\ast\bar\bH_i\bV^{-1} \bar\bH_i^\ast\bh_i \overset{(a)}{=}\tr\big(\bar\bH_i\bV^{-1} \bar\bH_i^\ast\bR_i\big),\overset{(b)}{=}\sum\limits_{j=1,j\neq i}^K 
\frac{\bh_j^\ast \bR_i \bh_j}{\bh_j^\ast\bh_j}\overset{(c)}{=}\sum\limits_{j=1,j\neq i}^K 
\frac{\tr(\bR_i\bR_j)}{\tr(\bR_j)},
\label{eq:twoRHS}
\end{align}
where $(a)$ is due to~\cite[Lemma 4]{Wagner2012Large}, $(b)$ is by simple algebraic manipulation, and $(c)$ is due to the use of~\cite[Lemma 4]{Wagner2012Large} in both the numerator and denominator. We obtain~\eqref{eq:ZFanaapp} by using~\eqref{eq:oneRHS} and~\eqref{eq:twoRHS} in~\eqref{eq:SINRZFsumform}.

To guarantee the non-negativity of~\eqref{eq:ZFanaapp}, it is sufficient to show that
\begin{align}
\tr(\bR_i)-\sum\limits_{j=1,j\neq i}^{K}\frac{\tr(\bR_i\bR_j)}{\tr(\bR_j)}\geq 0,\forall i.
\label{eq:ZFnonneg1}
\end{align}
If $\lambda_{\max}(\bR_i)$ represents the largest eigenvalue of $\bR_i$, then by using $\tr(\bR_i)=M,\forall i$ and re-arranging terms in \eqref{eq:ZFnonneg1}, we get
\begin{align}
M^2 \geq \sum\limits_{\substack{j=1 \\ j\neq i}}^{K}\tr(\bR_i\bR_j)\overset{(a)}{\geq}\sum\limits_{\substack{j=1 \\ j\neq i}}^{K}\lambda_{\max}(\bR_i)\tr(\bR_j)= M (K-1) \lambda_{\max}(\bR_i),
\end{align}
where $(a)$ is from the property $\tr(\bA\bB)\leq\lambda_{\max}(\bA)\tr(\bB)$~\cite{Fang1994Inequalities}. As $\lambda_{\max}(\bR_i)$ is $\cO(1)$ by A2, the above inequality is guaranteed to hold as $M/K\rightarrow\infty$. 
\vspace{-0.5em}\section*{Appendix B: Proof of Proposition 1}
We can write $\bar\bH_i^\ast \bar\bH=\bV+\bE$, where $\bV$ is a diagonal matrix with terms of $O(M)$, and the matrix $\bE$ has terms of $O(\sqrt{M})$. The first order Taylor series expansion of $(\bar\bH_i^\ast \bar\bH)^{-1}$ gives 
\begin{align}
(\bar\bH_i^\ast \bar\bH)^{-1}-\bV^{-1}\approx-\bV^{-1}\bE\bV^{-1}.
\end{align}
This result can be used in~\eqref{eq:epsilon} to write the error as 
\begin{align}
\epsilon&=| \bh_i^\ast \bar\bH_i (\bV^{-1}\bE\bV^{-1})\bar\bH_i^\ast \bh_i|.
\end{align}
The terms in the vector $\bar\bH_i^\ast \bh_i$ are $O(\sqrt{M})$. The terms in the matrix $\bV^{-1}\bE\bV^{-1}$ are order $O(\frac{1}{\sqrt{M}^3})$. Thus the terms in the product vector $\bV^{-1}\bE\bV^{-1} \bar\bH_i^\ast \bh_i$ are $O(\frac{\sqrt{K}}{M})$. Extending the same argument, the term $\bh_i^\ast \bar\bH_i (\bV^{-1}\bE\bV^{-1})\bar\bH_i^\ast \bh_i$ is $O(\frac{K}{\sqrt{M}})$. By using similar arguments, we can show that $\bh_i^\ast \bar\bH_i (\bar\bH_i^\ast \bar\bH_i)^{-1} \bar\bH_i^\ast \bh_i$ is $O(K)$.
\bibliographystyle{IEEEtran}
\bibliography{\centrallocation/Abbr,\centrallocation/Master_Bibliography}{}
\end{document}